\documentclass[10pt]{article}
\usepackage{longtable}
\usepackage{cite}
\usepackage{amsmath,amssymb,amsfonts}
\usepackage{amsthm}
\usepackage{subcaption}
\usepackage{caption}
\usepackage[noend]{algpseudocode}
\usepackage{algorithm}
\usepackage{fullpage}
\usepackage{graphicx}
\usepackage{multirow}
\usepackage[table,xcdraw]{xcolor}
\def\BibTeX{{\rm B\kern-.05em{\sc i\kern-.025em b}\kern-.08em
    T\kern-.1667em\lower.7ex\hbox{E}\kern-.125emX}}

\usepackage{comment}
\usepackage{mathtools}
\DeclarePairedDelimiter\ceil{\lceil}{\rceil}

\newtheorem{theorem}{Theorem}
\newtheorem{lemma}[theorem]{Lemma}
\newtheorem{proposition}[theorem]{Proposition}
\newtheorem{conjecture}[theorem]{Conjecture}

\begin{document}

\title{Optimal fermionic swap networks for Hubbard models}

\author{Tobias Hagge\footnote{Pacific Northwest National Laboratory, Richland, WA, USA. Email: tobias.hagge@pnnl.gov.}\;\footnote{Pacific Northwest National Laboratory is operated by Battelle Memorial Institute for the U.S. Department of Energy under Contract No. DE-AC05-76RL01830. This work was supported by the U.S. Department of Energy, under PNNL's QUASAR initiative. This material is based upon work supported by the U.S. Department of Energy, Office of Science, National Quantum Information Science Research Centers, Quantum Science Center.}}

\date{}
\maketitle
\thispagestyle{plain}
\pagestyle{plain}

\begin{abstract}

We propose an efficient variation of the fermionic swap network scheme used to efficiently simulate $n$-dimensional Fermi-Hubbard-model Hamiltonians encoded using the Jordan-Wigner transform. For the two-dimensional versions, we show that our choices minimize swap depth and number of Hamiltonian interaction layers. The proofs, along with the choice of swap network, rely on isoperimetric inequality results from the combinatorics literature, and are closely related to graph bandwidth problems. The machinery has the potential to be extended to maximize swap network efficiency for other types of lattices.
\end{abstract}

\section{Introduction}
\label{sec:intro}

Near-term quantum computing on limited and noisy hardware requires computations of small circuit depth, using few qubits. In this setting, algorithms with low overhead may be preferable to those with the best possible asymptotic scaling. On the other hand, algorithms with qubit connectivity requirements that are a good match for hardware have an advantage over those that require extra circuitry to account for qubit connectivity, as well as those that cannot take advantage of hardware capabilities.

Many problems of near-term interest in quantum chemistry require representation of a second-quantized fermionic Hamiltonian $H$ in terms of qubit operators. Jordan-Wigner encodings are the simplest means of accomplishing this purpose, but require establishing a canonical ordering of the fermion occupancy sites. Aside from a few considerations involving ancilla qubits, such schemes only require, and in some sense can only take advantage of, linear qubit connectivity. Such schemes remain competitive for near-term applications because of their extremely low overhead. The results of this paper highlight that Jordan-Wigner encodings have another advantage, at least when it comes to analysis: linear orders are well-studied and relatively easy to understand. In particular, we show that they can be used to construct circuits for simulation of fermionic lattice-model Hamiltonians which are optimal in a couple of senses.

The nickel version of our scheme, for 2-dimensional spin and spinless Fermi-Hubbard models is as follows: rather than interleave the rows of the model as is shown in \cite{kivlichan2018quantum}, and as is the usual practice, it is more efficient to interleave the diagonals. We prove that this is the best possible network. There are $n$-dimensional generalizations of our swap network, but it is more difficult to prove optimality and local Hamiltonian encoding methods seem likely to be better performers. Our proof is interesting because it can be applied to other lattices and because there is an interesting generalization which could take advantage of better-than-linear qubit connectivity.

\section{Background}
\label{sec:background}

The problem of Hamiltonian simulation was the first compelling quantum computing application \cite{Feynman1982}. Fermionic Hamiltonians, in particular, have important applications in quantum chemistry and materials science. Hamiltonians for Fermionic systems require special computational treatment, because in both first and second quantized formulations, interactions which respect Fermionic symmetries incur additional cost.

The second-quantized formulation has received more attention in the quantum computing literature, and is preferred when the number of Fermions is large relative to the number of occupancy sites. In this formulation, each Fermion occupancy site corresponds to a tensor factor in the state space; spin $\frac{1}{2}$ Fermions can be represented locally in qubit-space as single qubits. Fermionic statistics are preserved by mapping the Fermionic algebra operators to operators in qubit-space via the Jordan-Wigner transform. This fixes an ordering of the occupancy sites and destroys locality; a Fermionic operator on sites $p$ and $q$ becomes a qubit-space operator on qubits $p$ through $q$.

Many second-quantized Hamiltonian evolution methods employ Trotter-Suzuki decomposition in a process colloquially known as {Trotterization}. In this setting, one has a  {\em local $k$-body Hamiltonian}, i.e. a Hamiltonian which is expressible as a sum of polynomially many $k$-local summands, each of which evolves at most $k$ particles via creation and annihilation operators among at most $2k$ sites. The evolution of this Hamiltonian is approximated, on a short time scale, as a product of short-time evolutions of the local Hamiltonian summands. One body Hamiltonians, for example, are of the form:
\[H = \sum_i H_i + \sum_{i,j} H_{i,j}, \mbox{ where}\]
\[H_i = k_i a^\dagger_i a_i,\]
\[H_{i,j} = k_{i,j} a^\dagger_i a_j,\]

with the $k$'s constant and the $a_i$ and $a_i^\dagger$ being fermionic annihilation and creation operators, respectively. Using the first-order Trotter-Suzuki formula, a Trotter step for such a Hamiltonian is of the form:

\[e^{-i H \Delta t} \cong \prod_j e^{-i H_j \Delta t} \times \prod_{j,k} e^{-i H_{j,k} \Delta t}.\]
Computing a Trotter step efficiently necessitates efficient circuits for the short-time evolutions.

In the first published algorithm for fermionic Hamiltonian simulation, due to Abrams and Lloyd \cite{abrams-lloyd-1997}, each $H_{i,j}$ operator on $n$ qubits required $O(n)$ circuit-depth, due to the cost of implementing fermionic opertors, with arbitrary pairwise qubit connectivity. The Bravyi-Kitaev method \cite{bravyi-kitaev} reduces the asymptotic circuit-depth to $O(\log n)$ by introducing linearly-many ancilla qubits to track the fermionic operator sign sums over contiguous ranges of qubits, and is more efficient in practice \cite{seeley2012bravyi-kitaev}. The cost analysis states no connectivity restrictions, and in practice requires $O(1)$ two-qubit (bosonic) operations over a binary tree on the ancilla qubits.

Several methods decrease asymptotic circuit complexity by attempting to improve on qubit connectivity, either by replacing the Hamiltonian with a local Hamiltonian with the same ground state (the Verstraete-Cirac-Ball method \cite{Verstraete_2005} \cite{ball}), or by directly encoding enough of the fermionic operator algebra to represent the Hamiltonian's one-body interactions (e.g. the Bravyi-Kitaev Superfast Method \cite{bravyi-kitaev} and it's Generalized Superfast Encodings variant \cite{gse}, the Auxilliary Qubit method, and compact fermion to qubit encodings \cite{compact}). These methods make one-body interactions local at the cost of some qubit and circuit-complexity overhead, assuming that the qubit connectivity supports the locality of these operations.

The fermionic swap method \cite{kivlichan2018quantum} is an attempt to make the Jordan-Wigner encoding as cheap as possible by amortizing the cost of the interaction operators. It works by efficiently transposing qubits adjacent in the Jordan-Wigner ordering so that each interacting pair of qubits, or for multi-body Hamiltonians, each interacting tuple of qubits, is brought to adjacency by fermionic swap operators. These are swap operators modified to preserve Fermi exchange statistics; once brought to adjacency the local fermionic operators become local qubit operators and may be performed in constant time. Fermionic swap operators have low overhead, do not require auxilliary qubits, and the linear qubit connectivity requirements are quite modest.

The cost of the Hamiltonian interaction layers depends greatly on the underlying hardware. If the hardware can perform arbitrary one-qubit gates, it may be possible to absorb some or all of the cost of interactions into the swap layers. Otherwise, while the Kitaev-Solovay theorem \cite{Kitaev_1997} guarantees that accurate gates can be obtained cheaply, interaction layers will usually be expensive compared to swap layers.

For $n$-occupancy-site one-body fermionic Hamiltonians which are dense (i.e. all $k_{i,j} \ne 0$), the required $n(n-1)$ pairwise interactions (among $\frac{n(n-1)}{2}$ interaction pairs) are accomplished with $n-2$ layers of swaps and $n$ interaction layers.\footnote{As presented in \cite{kivlichan2018quantum}, each qubit is swapped with every other qubit, using $n$ swap layers. However, qubits need only be made adjacent; inspection shows that the $n$-th swap layer is unnecessary, as the $n-1$st since the interactions in the final layer are available prior to the start of swapping. We treat these trivial optimizations as part of the original method.} The authors conjecture that no other method can perform a Trotter step for such Hamiltonians using fewer entangling gates.

To motivate what follows, the following are the fermionic swap method's optimality properties with respect to the number of swap layers (swap depth) and Hamiltonian interaction layers (Hamiltonian interaction depth) for the fermionic swap method for dense one-body Hamiltonians:
\begin{lemma}
For dense one-body Hamiltonians, the swapping scheme of \cite{kivlichan2018quantum} has optimal swap depth. For $n$ odd, it has optimal interaction depth. For $n$ even, $n>2$, the interaction depth is one greater than optimal; swap-depth and interaction-depth optimality may not both be achieved simultaneously.
\end{lemma}
\begin{proof}
Given a swap network for a dense one-body Hamiltonian, follow the least qubit, in the initial ordering, as it swaps with the others. It is lower-ordered than every other qubit at the time it first becomes adjacent to it. Thus the swap depth is at least $n-2$, and \cite{kivlichan2018quantum} has optimal swap depth.

The interaction depth is bounded below by the degree of the Hamiltonian interaction graph, which in this case is $n-1$. To attain this bound, each qubit must interact with another in each of the interaction layers. If $n$ is odd this is impossible. If $n$ is even, by a simple counting argument, each layer contains $\frac n 2$ interactions. This implies that each pair of qubits $(2k, 2k+1)$ must interact within each interaction layer. Suppose additionally that the network were swap-depth optimal. Then each pair of successive interaction layers would be separated by a single swap layer which places new qubit pairs in each of the $(2k, 2k+1)$ slots. Each such sway layer would need to swap the pair $(1,2)$, but $0$ can interact only with $1$ and $2$ in this way. Thus, if $n>2$ and even, no interaction-depth optimal network can be swap-depth optimal.

Attaining interaction depth $n-1$ for $n$ even amounts to constructing an $n-1$-color edge-coloring on the complete graph on $n$ vertices, which is possible by Baranyai's theorem \cite{baranyai}, and assigning the edges of one color to a single interaction layer.

\end{proof}

\newcommand{\swapdepth}[1]{d_{\operatorname{swap}}(#1)}
\section{Optimal swap networks for Hubbard models}

\subsection{Preliminaries}

By the term {\em grid graph}, we mean a finite graph product $G_1 \times \ldots \times G_k$ of path graphs $G_i$. In particular, for $M$ and $N$ positive integers, $M \le N$, the $M \times N$ grid graph $G = (V, E)$  with vertices $V$ and edges $E$ is defined as follows:
\begin{equation}
\begin{array}{lll}
V & = & \{(m,n) | 0 \le m < M, 0 \le n < N\},\\
E_h & = & \{((m,n), (m+1,n)) | 0 \le m < M-1, 0 \le n < N\},\\
E_v & = & \{((m,n), (m,n+1)) | 0 \le m < M, 0 \le n < N-1\},\\
E & = & E_h \cup E_v.
\end{array}
\end{equation}

In the spinless $M \times N$ Hubbard model, there is one occupancy state for each element of $V$, and the Hamiltonian is given by:
\[H = U \sum_{p \in V} n_p - t \sum_{(p,q) \in E} (a_p^\dagger a_q + a_q^\dagger a_p). \]
Here, $a_p$ and $a_p\dagger$ are the $p$th fermionic operator algebra raising and lowering operator, respectively, $n_p$ is the $p$th number operator,$-t$ is the kinetic energy, and $U$ is the Coulomb repulsion.

For the version with spin, each element of $V$ is assigned two possible occupancy states, one for each electron spin. The Hamiltonian is then
\[H_s = U \sum_{p \in V} n_{p, +} n_{p,-} - t \sum_{(p,q) \in E, \sigma \in \pm} (a_{p,\sigma}^\dagger a_{q,\sigma} + a_{q,\sigma}^\dagger a_{p,\sigma}). \]
The graph of two-occupancy-state interactions for this Hamiltonian is the $2 \times M \times N$ grid graph.

The swap network is concerned with the sites which must be brought to adjacency in order to compute the Hamiltonian interaction terms. The form of those terms determines the structure of the Hamiltonian interaction layers.

Swap networks for two-dimensional Hubbard Hamiltonians were considered in \cite{kivlichan2018quantum}; the method therein attains $\frac{3(M-1)}{2}$ swap layers to process a single Trotter step in the spinless case, and $\frac{3(2M-1)}{2}$ layers in the case with spin. The networks in \cite{kivlichan2018quantum} contain $O(M)$ Hamiltonian interaction layers. In practice, the first $\frac{M-1}{2}$ (respectively, $\frac{2M-1}{2}$) layers can be omitted, giving $M-1$ and $2M-1$ layers respectively. Note that there is there is no hope of obtaining the minimum number of entangling gates, for this or any other swap network, as the local Hamiltonian methods have better scaling.

The degree of a graph is a lower bound on the number of Hamiltonian interaction layers in a swap network. This gives lower interaction-depth bounds of $4$ and $5$  for two-dimensional spinless and spin Hubbard models respectively.

\subsection{Optimal swap depths}

Given a set $V$, let $R(V)$ denote the set of all orders of $V$, that is the set of bijective functions $V \to \{0, \ldots, |V| - 1\}$. Given an order $r$ of $V$, let $r_k = r^{-1}(\{0, \ldots, k-1\})$, the $k$-th initial segment of $r$.

Given
a graph $G = (V,E)$ , let $\swapdepth{G}$ be the minimum swap depth over all linear-connectivity swap networks for $G$, and define the {\em graph bandwidth} $b(G)$ of $G$ as
\[b(G) = \min_{r \in R(V)} \max_{(v,w) \in E} |r(v) - r(w)|.\]

Graph bandwidth determines a simple lower bound on swap depth:
\begin{equation}\label{bandwidthbound}
\swapdepth{G} \ge \ceil{\frac{b(G)-1}{2}},
\end{equation}
since for any $r \in R(V)$, $G$ has an edge requiring at least $\ceil{\frac{b(G)-1}{2}}$ overlapping swaps (and thus, swap layers) to bring the vertices to adjacency.

In general, computing graph bandwidth is an NP-hard problem. For $G$ an $M \times N$ grid graph, with $M \le N$, it is shown in \cite{Chvatalova1975249} that $b(G) = M$, giving $\swapdepth{G} \ge r\ceil{\frac{M-1}{2}}$. For $G$ an $2 \times M \times N$ grid graph, with $2 \le M \le N$, it follows from \cite{wangwang77} that $b(G) = 2 M - 1$, giving  $\swapdepth{G} \ge M - 1$ .

The graph bandwidth can also be expressed as
\[b(G) = \min_{r \in R(V)}\max_{e \in E} b_r(e),\]
where for any $W \subset V$, the {\em order bandwidth} $b_r(W)$ of $W$ under $r$ is given by
\[b_r(W) = \max_{w \in W} r(w) - \min_{w \in W} r(w)\]

We generalize this slightly; let $L_2(G)$ be the collection of all subgraphs $S = (V_S, E_S)$ of $G$ which are length-two line graphs. Define the {\em graph 2-bandwidth} $b^2(G)$ as follows:
\[b^2(G) = \min_{r \in R(V)} \max_{S \in L_2(G)} b_r(V_S).\]

\begin{lemma}\label{twobandwidthbound}
\[\swapdepth{G} \ge \ceil{\frac{b^2(G) - 2}{2}}.\]
\end{lemma}
\begin{proof}
Given an order $r$ of $V$, and a length-two line-subgraph $S$ with vertices $v_1, v_2, v_3$, $r(\{v_1, v_2, v_3\}) = \{p_1,p_2,p_3\}$ for some nonnegative integers $p_1 < p_2 < p_3$. Then $b_r(S) = p_3 - p_1$. The swap depth cost to bring both edges of $S$ to adjacency (not necessarily simultaneously) is at least $\ceil{\frac{p_3 - p_1 - 2}{2}}$, as if $r(v_2) = p_2$ it cannot be brought nearer to both $v_1$ and $v_3$ by a single swap, and otherwise it requires $\ceil{\frac{p_3 - p_1 - 1}{2}}$ swaps to bring the vertices at $p_1$ and $p_3$ to adjacency. For at least one such $S$, $p_3 - p_1 \ge b^2(G)$.
\end{proof}

As we shall explain, for many parameterized graph families for which exact bandwidth results are known, the bandwidth is provably realized by a vertex order with special boundary-optimality properties.

Given a graph $G = (V, E)$ and $W \subset V$, the {\em vertex boundary} $B_G(W)$ of $W$ in $G$ is the set of vertices in $V - W$ which are adjacent to vertices in $W$. The {\em vertex boundary closure} $C_G(W)$ is $W \cup B_G(W)$.

For any $r \in R(V)$ and any initial segment $r_k$, the highest-ordered element of $B_g(r_k)$ shares an edge with some element of $r_k$. Thus
\[\max_{k' \le |V|, r^{-1}(k'-1) \in B_G(r_k)}k' \ge k + |B_G(r_k)|\]
and therefore
\begin{equation}\label{boundary_bandwidth_inequality}
    b(G) \ge \min_{r \in R(V)} \max_{k \le |V|} |B_G(r_k)|.
\end{equation}

Let $G=(V,E)$ be a graph, $r \in R(V)$. We say $r$ is an {\em isoperimetric order} on $V$ for $G$ if for any other order $r'$ of $V$ and $0 \le k /le |V|$, the $k$-th initial segments $r_k$ and $r'_k$ of $r$ and $r'$ satisfy $|B_G(r_k)| \le |B_G(r'_k)|$. 

It is often possible to construct an isoperimetric order $r$ with the property of {\em initial-segment closure}: for any initial segment $r_k$, $C_G(r_k)$ is also initial segment.
\begin{lemma}\label{lmm:two_bandwidth}
In Inequality \ref{boundary_bandwidth_inequality}, if $r$ is an isoperimetric order with initial-segment closure, equality holds.
\end{lemma}
\begin{proof}
For any $0 \le k \le |V|$, let $k' = k + |B_G(r_k)|$. Then $C_G(r_k) = r_{k'}$, and if $k>0$, any element of $B_G(r_k)$ adjacent to $r_k$'s largest element, $r^{-1}(k-1)$, is equal to $r^{-1}(k'')$ for some $k'' \le (k-1) + |B_G(r_k)|$. Thus the order bandwidth of every edge is bounded above by some $|B_G(r_{k''})|$. Given any other order $r'$, $|B_G(r'_{k''})|\ge | B_G(r_{k''})|$ and thus equality holds. 
\end{proof}

Similar arguments show that in general
\begin{equation}\label{two_bandwidth_inequality}
    b^2(G) \ge \min_{r \in R(V)} \max_{k \le v} | B_G(r_k) \cup B_G(C_G(r_k)) |.
\end{equation}
\begin{lemma}
In Inequality~\ref{two_bandwidth_inequality}, if $r$ has initial-segment closure, equality holds.
\end{lemma}
\begin{proof}
The proof is similar to that of Lemma~\ref{lmm:two_bandwidth}; the only additional complexity is that the lowest-ordered vertex $v$ is the middle vertex for some line subgraphs, however, setting $k=r(v)+1$ the other two vertices lie in $r_{k + |B_G(r_k)|}$ and thus in $r_{k + |B_G(r_k) + B_G(C_G(r_k))|}$.
\end{proof}

\begin{proposition}[\cite{wangwang77}]
Let $0 \le M_1 \le ... \le M_n$, with each $M_i$ an integer. Let $G$ be the $M_1 \times \ldots \times M_n$ grid graph.
Let $r^*$ be the order given by
\[(x_1, \ldots, x_n) < (y_1, \ldots, y_n)\]
if either
\[\sum x_i < \sum y_i\]
or
\[\sum x_i = \sum y_i \text{ and } (x_1, \ldots, x_n) >_L (y_1, \ldots, y_n),\]
where $L$ is the lexicographic order. Then $r^*$ is an isoperimetric order on $G$.
\end{proposition}

In addition, it is easily verified that $r^*$ has initial-segment closure. Many other graphs are known to admit isoperimetric orderings; see \cite{harper_2004} for an overview of what is known.

For any connected graph $G = (V,E)$ with order $r$ on $V$ we can partition $V$ into sets $V_i$ as follows:
\[\begin{array}{l}V_0 = r^{-1}(\{0\}), \\
V_i = B_G(V_0 \cup \ldots \cup V_{i-1})
\end{array}
\]

If $G$ is bipartite, no edge in $E$ can have both endpoints in the same $V_i$ (or else there would be two paths from $0$ of opposite parity), and conversely, if no $V_i$ contains the endpoints of an edge in $E$, the even and odd $V_i$ form a bipartite partition of $V$.

For a grid graph, every vertex $v$ except the first and the last under $r^*$ has edges to smaller vertices and to greater vertices. The length-two path $P$ through $v$ which maximizes order bandwidth $b_{r^*}(P)$ is that from the least neighbor of $v$ to the greatest.

For the case that $G$ is an $M \times N$ grid graph, $M \le N$, inspection reveals that $|B_G(r^*_k)| = 2$ when $k = 1$, increases by one as the first element $(i,0)$ in each $V_i$ is reached, for $i \le M-2$, at which the point at $(M-2,0)$ the maximum is attained. $|B_G(r^*_k)|$ then remains constant until the last element of $V_{N-1}$, $(0, N-1)$ is filled, at which point it decreases by one upon completion of each $V_i$. Since $M \le N$, $|B_G(r^*_k)|$ remains maximal for at least $2M - 1$ steps. Applying Lemma~\ref{twobandwidthbound}, we obtain the following:
\begin{equation}\label{eq:spinless-swap-bound}
    \begin{array}{l}
     b^2(G) = 2M, \\
     \swapdepth{G} \ge M-1.
     \end{array}
\end{equation}

For higher-dimensional grid graph cases, these quantities can be computed using the structure of the $V_i$. The set of points
\[\{(x_1, \ldots, x_n)| x_i \in \mathbb N, \sum x_i = k\}\]
forms an {\em $(n-1)$-simplicial grid of size $k$} (when $k=1$ it is just the standard $(n-1)$-simplex). Such coordinates might not all be valid on our grid graph; for each $M_i$ one must remove the $(n-1)$-simplicial grid of size $x_i - M_i$, with the convention that grids of negative size are empty, consisting of points not lying on the grid graph because $x_i$ is too large.

In the $2 \times M \times N$ case, $r^*$ partitions $V$ into the $V_i$, with the spin-coordinate-one vertices in each $V_i$ prior to all of the the spin-coordinate-zero vertices. It is easy to see that $|B_G(r^*)|$ increases monotonically until reaching a maximum of $2M$, upon reaching $(1, M-2, 0)$ in $V_{M-1}$. This maximum is maintained until reaching $(0, 0, N-1)$, in $2M - 1$ more steps if $M = N$, greater than $2M$ steps otherwise, at which point $|B_G(r^*)|$ decreases by one and remains stable for at least $M-1$ steps.

Applying Lemma~\ref{twobandwidthbound}, we obtain the following:
\begin{equation}\label{eq:spin-swap-bound}
    \begin{array}{l}
     b^2(G) = \left\{
        \begin{array}{ll}
            4M & \text{if } M <N, \\
            4M - 1 & \text{if } M=N.
        \end{array}
    \right.\\
     \swapdepth{G} \ge 2M-1.
     \end{array}
\end{equation}

\subsection{Hubbard model networks}

Our strategy for swap networks on grid graphs is as follows. Suppose $(G,E)$ is a grid graph of dimension $n$.  Each $V_i$ consists of a concatenated sequence of sequences, which we shall call $\sigma$-rows, each of the form:
\[\sigma^i_{(x_1, \ldots, x_{n-2})} = (k, 0, x_1, \ldots, x_{n-2}), (k-1, 1, x_1, \ldots, x_{n-2}) \ldots, (0, k, x_1, \ldots, x_{n-2}).\]
For each $\sigma$-row $\sigma^i_{(x_1, x_2, \ldots, x_{n-2})}$, the $k$-th element must interact with the $k$-th element of each of the following sequences for $j \in \pm 1$, provided it exists:
\[\sigma^{i+j}_{(x_1 + j, x_2, \ldots, x_{n-2})}, \sigma^{i+j}_{(x_1, x_2 + j, \ldots, x_{n-2})}, \ldots \sigma^{i+j}_{(x_1, x_2, \ldots, x_{n-2} + j)}.\]
Additionally, for $j \in \pm 1$, $\sigma^i_{(x_1, x_2, \ldots, x_{n-2})}$ must interact with the $k$-th and $k+j$-th elements of each $\sigma^{i+j}_{(x_1, x_2, \ldots, x_{n-2})}$, whenever one exists.

The difference in length between the original $\sigma$-row and any of of its neighbor rows is at most one.

Let $r^o$ be the order on the bipartite partition $\bigcup V_{2i+1}$ which places the $V_{2i+1}$ in ascending order, each ordered as a sequence of $\sigma$-rows, in reverse lexicographical order by index. Define $r^e$ similarly on $\bigcup V_{2i}$.

 Interlace the two orders $r^o$ and $r^e$ so that the elements of $r^o$ have the rightmost possible placements such that elements of $r^o$ maintain relative order and each element of $\sigma^{2i+1}_{(x_1, \ldots, x_{n-2})}$ is adjacent to its (up to) two neighbors in $\sigma^{2i}_{(x_1, \ldots, x_{n-2})}$, if it exists. An interaction layer is performed. The $\sigma^{2i+1}_{(\ldots)}$ are then shifted upward in $2n - 1$ stages to be brought adjacent to each of the remaining neighbor $\sigma$-rows, performing an interaction layer at each stage.

 Our lower bounds on swap depth are not strict for higher-dimensional grid graphs. In general, not every vertex in $r^o$ can be shifted relative to $r^e$ at each time step. When elements of $r^o$ are adjacent to each other they must queue up to swap with the next element of $r^e$, increasing the transit time for a given $v \in r^o$. While the approach seems fairly efficient, it is not clear that it is optimal and as the dimension increases, local Hamiltonian methods become more attractive.

\begin{theorem}

For the spinless and spin $M \times N$ Hubbard model Hamiltonians, the above swap-depth bounds are realizable by the above swap network scheme. In the spinless case, interaction-depth optimality is simultaneously achievable. In the spin case, one extra interaction layer is required.
\end{theorem}

\begin{proof}
For the spinless case, each $V_i$ contains a single $\sigma$-row. It requires at most $M-1$ swaps on each vertex to shift each of the $V_{2i+1}$ from its initial to its final position. There may at most one extra stage due to queueing between the rows if $|V_{2i+3}| > |V_{2i+2}|$, however in this case the number of swaps is smaller than $M-1$. Thus the required swap depth is $M-1$. By Equation~\ref{eq:spinless-swap-bound}, the network is swap-depth optimal. See Figure~\ref{fig:hubbard-order} for an example.

For the Hamiltonian with spin, each $V_i$ consists of $\sigma$-rows of length at most two. All but possibly the first and last rows have length exactly two, so gaps can form only at the beginning and end  of each $V_i$. Each $\sigma$-row interacts with two others, which are adjacent to each other in their bipartite ordering. As a result, to perform the stages, the $V_i$ are just shifted as rows without any queueing within the $V_i$ (see Figure~\ref{fig:spin-hubbard-order} for an illustration). Such a shift requires at most $2M - 1$ swaps, with a single queue step occurring only when fewer than $2M - 1$ swaps occurred. Thus the required swap depth is $2M - 1$. By Equation~\ref{eq:spin-swap-bound}, the network is swap-depth optimal. The network minimizes three-dimensional grid interaction depth, which is six, but since the grid has $M_1 = 2$, the degree of each vertex is only five.

\end{proof} 

\begin{figure}
    \centering
    \includegraphics[width=5in]{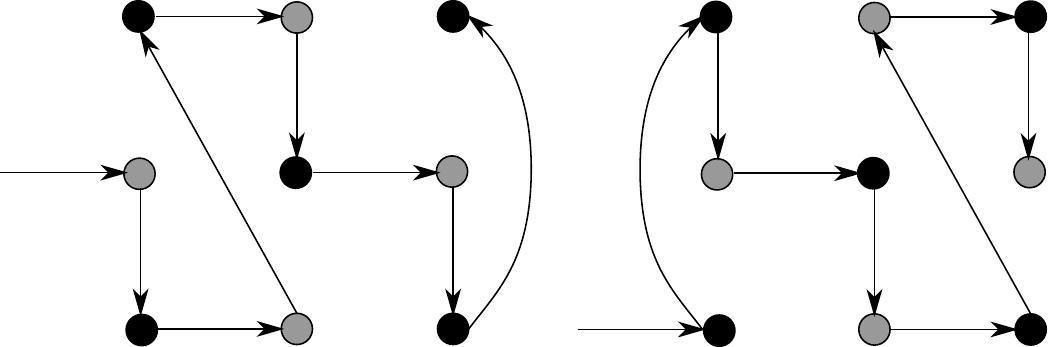}
    \caption{The initial (left) and final (right) orders for a $3 \times 3$ spinless Hubbard model.}
    \label{fig:hubbard-order}
\end{figure}

\begin{figure}
    \centering
    \includegraphics[width=3in]{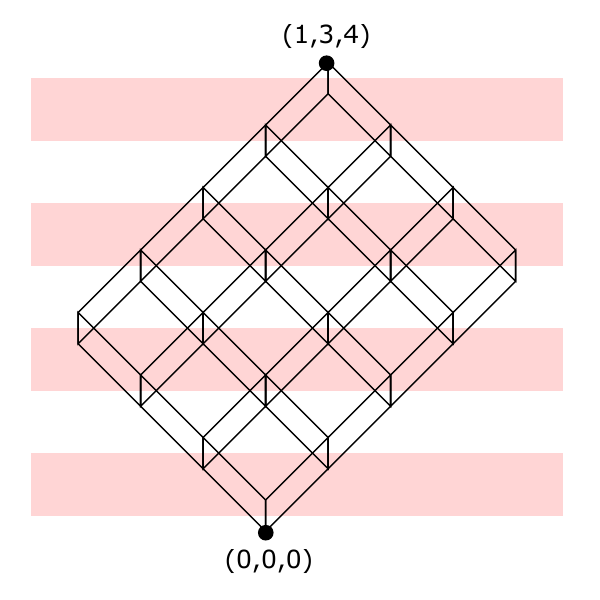}
    \caption{An isoperimetric ordering for the Hubbard model is given by ordering the above vertices bottom to top, left to right. For odd $i$, the $V_i$ are shaded. Within each $V_i$, the order is reverse-lexicographic.}
    \label{fig:spin-hubbard-order}
\end{figure}

\section{Generalizations}
The above machinery applies straightforwardly to any bipartite graph which admits an isoperimetric ordering. Since isoperimetric orders are orders attainable using a greedy strategy, it is often easy to guess them for lattice graphs, and there is also a literature (see e.g. \cite{global_isoperimetric}). For lattice graphs which are not bipartite we are not guaranteed that interacting pairs within a single $V_i$ will be brought to adjacency, however, it sometimes happens. For example, if, in the spinless Hubbard model, using an isoperimetric ordering, edges are added between the adjacent elements of each $V_i$, one obtains a triangular lattice. These new edges come to adjacency in the spinless Hubbard swap network, and interactions could be performed. The result has minimal swap depth (the same proof goes through) but not minimal interaction depth.

\begin{conjecture}
  Given a bipartite interaction graph with an isoperimetric ordering, interleaving the bipitartite layers as above produces a swap network with the minimum number of entangling gates.
\end{conjecture}

There is another potential interesting generalization to consider. At a certain scaling, local Hamiltonian methods become more efficient than fermionic swap methods, assuming quantum hardware supports the local fermionic Hamiltonian connectivity. However, real-hardware connectivity is limited and it may not be feasible to build custom hardware for each Hamiltonian interaction graph to be simulated. In this circumstance, it makes sense for locality of local fermionic Hamiltonian methods to match the connectivity of the hardware, rather than the connectivity of the Hamiltonian. This makes the Hamiltonian nonlocal, but it is still ``more local'' than it would be under the Jordan-Wigner encoding.

Suppose for example the hardware admits local fermionic operations on a square lattice. The graph bandwidth problem can be thought of as finding the vertex mapping from the Hamiltonian interaction graph to a path graph which minimizes maximal graph distance between Hamiltonian-graph-adjacent vertices. If the path graph is replaced with the square lattice, the graph bandwidth problem becomes the {\em quadratic bottleneck assignment problem}. Like the graph bandwidth problem, it is NP-hard, and the problem needs to be generalized (just as the graph bandwidth problem does) in order to discover optimal swap networks. Nevertheless, there may be cases with enough structure that a provably optimum solution is possible.

\bibliographystyle{abbrv}
\bibliography{bib/isoperimetric,bib/tobe}

\end{document}